\documentclass{llncs}

\usepackage{enumitem}
\setlist[itemize]{noitemsep, topsep=0pt}
\usepackage{cite}
\usepackage{xspace}
\usepackage[a4paper]{geometry}

\usepackage{url}
\usepackage{amsmath}
\usepackage{epsfig}
\usepackage{fancyvrb}
\usepackage{listings}
\usepackage{etex} 
\usepackage{stmaryrd} 

\usepackage{graphicx}
\usepackage{proof}

\usepackage{alltt}
\usepackage{times}

\newcommand{\reffig}[1]{Fig.~\ref{#1}}

\newcommand{\refsec}[1]{Section~\ref{#1}}

\usepackage{oz}
\usepackage{zed-csp}

\usepackage{color}

{\begin{equation}\begin{minipage}{0.85\linewidth}\raggedright\begin{zed}}%
{\end{zed}\end{minipage}\end{equation}}
\newenvironment{fma*}%
{\begin{equation*}\begin{minipage}{0.85\linewidth}\raggedright\begin{zed}}%
{\end{zed}\end{minipage}\end{equation*}}

\usepackage{xifthen}

\def\nat        {\hbox to 1.3pt{{I}\hss}\mbox{N}}
\def\func #1 #2 { $f$:\mbox{#1}$\rightarrow$\mbox{#2} }

\newcommand{\linefill}{\cleaders\hbox{$\smash{\mkern-2mu\mathord-\mkern-2mu}$}\hfill\vphantom{\lower1pt\hbox{$\rightarrow$}}}  
  
\newcommand{\transi}[2]{\mathrel{\lower1pt\hbox{$\mathrel-_{\vphantom{#2}}\mkern-8mu\stackrel{#1}{\linefill_{\vphantom{#2}}}\mkern-11mu\rightarrow_{#2}$}}}

\newcommand{\ntransi}[2]{\mathrel{\lower1pt\hbox{$\mathrel-_{\vphantom{#2}}\mkern-8mu\stackrel{#1}{\linefill_{\vphantom{#2}}}\mkern-8mu\nrightarrow_{#2}$}}}

\newcommand{\abort}{\texttt{abort}}

\newcommand{\tmbegin}{\texttt{TMBegin}}
\newcommand{\tmbegino}{\texttt{TMBegin}({\tt ok})}
\newcommand{\tmend}{\texttt{TMCommit}}
\newcommand{\tmendc}{\texttt{TMCommit}({\tt ok})}
\newcommand{\tmenda}{\texttt{TMCommit}({\tt abort})}
\newcommand{\tmread}{\texttt{TMRd}}
\newcommand{\tmreada}{\texttt{TMRd}({\tt abort})}
\newcommand{\tmwrite}{\texttt{TMWr}}
\newcommand{\tmwriteo}{\texttt{TMWr}({\tt ok})}
\newcommand{\tmwritea}{\texttt{TMWr}({\tt abort})}

\definecolor{darkgreen}{rgb}{0,0.2,0}

\usepackage{inconsolata}
\usepackage[T1]{fontenc}

\renewcommand{\pfun}{\p\tfun} 

\newcommand{\ops}{{\it ops}}

\newcommand{\Tids}{T}

\newcommand{\crash}{\mathit{c}}

\newcommand{\beginInv}[1]{\ensuremath{inv_{#1}(\tmbegin)}}
\newcommand{\beginResp}[1]{\ensuremath{resp_{#1}(\tmbegin)}}
\newcommand{\readInv}[2]{\ensuremath{inv_{#1}(\tmread(#2))}}
\newcommand{\readResp}[2]{\ensuremath{resp_{#1}(\tmread(#2))}}
\newcommand{\writeInv}[2]{\ensuremath{inv_{#1}(\tmwrite(#2))}}
\newcommand{\writeResp}[1]{\ensuremath{resp_{#1}(\tmwrite)}}
\newcommand{\commitInv}[1]{\ensuremath{inv_{#1}(\tmend)}}
\newcommand{\commitResp}[1]{\ensuremath{resp_{#1}(\tmend)}}

\newcommand{\abortResp}[1]{\ensuremath{resp_{#1}(\abort)}}
\newcommand{\run}[1]{run_t}
\newcommand{\crashed}{crash}

\newcommand{\doCommitReadOnly}[2]{\ensuremath{\texttt{DoCommitReadOnly}_{#1}(#2)}}
\newcommand{\doCommitWriter}[1]{\ensuremath{\texttt{DoCommitWriter}_{#1}}}
\newcommand{\doRead}[2]{\ensuremath{\texttt{DoRead}_{#1}(#2)}}
\newcommand{\doWrite}[2]{\ensuremath{\texttt{DoWrite}_{#1}(#2)}}

\newcommand{\action}[3]{\ensuremath{
\begin{array}[t]{ll}
\multicolumn{2}{l}{#1}\\
\textsf{Pre: }&#2\\
\textsf{Eff: }&#3
\end{array}
}}

\newcommand{\statusText}[1]{\text{#1}}

\newcommand{\pcNotStarted}{\statusText{notStarted}}
\newcommand{\pcBeginPending}{\statusText{beginPending}}

\newcommand{\pcReady}{\statusText{ready}}
\newcommand{\pcDoWrite}{\statusText{doWrite}}
\newcommand{\pcWriteResp}{\statusText{writeResp}}
\newcommand{\pcDoRead}{\statusText{doRead}}
\newcommand{\pcReadResp}{\statusText{readResp}}
\newcommand{\pcDoCommit}{\statusText{doCommit}}
\newcommand{\pcCommitResp}{\statusText{commitResp}}

\newcommand{\pcCommitted}{\statusText{committed}}
\newcommand{\pcAborted}{\statusText{aborted}}

\newcommand{\sif}{\texttt{if}}
\newcommand{\sthen}{\texttt{then}}
\newcommand{\selse}{\texttt{else}}

\newcommand{\DTMS}{\text{\sc dTMS2}\xspace}
\newcommand{\DTML}{\text{\sc dTML}\xspace}
\newcommand{\TMS}{\text{\sc TMS2}\xspace}
\newcommand{\TML}{\text{\sc TML}\xspace}

\def \st {\spot}

\title{Defining and Verifying Durable Opacity: Correctness for
  Persistent Software Transactional Memory \thanks{Bila and Dongol are
    supported by VeTSS project ``Persistent Safety and Security''.
    Dongol is supported by EPSRC projects EP/R019045/2 and
    EP/R032556/1. Derrick and Doherty are supported by EPSRC grant
    EP/R032351/1. Wehrheim is supported by DFG grant \mbox{WE2290/12-1}. }}

\author{Eleni Bila$^1$ \and Simon Doherty$^2$ \and Brijesh
  Dongol$^1$  \and John Derrick$^2$ \and 
  Gerhard Schellhorn$^3$ \and Heike Wehrheim$^4$}

\institute{University of Surrey, UK \and University of Sheffield, UK \and University of Augsburg, Germany \and Paderborn University, Germany}

\begin{document}
\label{firstpage}

\maketitle

\begin{abstract}

  Non-volatile memory (NVM), aka persistent memory, is
  a new para\-digm for memory that preserves its contents even after
  power loss. The expected ubiquity of NVM has stimulated interest in
  the design of novel concepts ensuring correctness of concurrent
  programming abstractions in the face of persistency. So far, this
  has lead to the design of a number of persistent concurrent data
  structures, built to satisfy an associated notion of correctness:
  durable linearizability.

  In this paper, we transfer the principle of durable concurrent
  correctness to the area of software transactional memory
  (STM). Software transactional memory algorithms allow for concurrent
  access to shared state. Like linearizability for concurrent data
  structures, opacity is the established notion of correctness for
  STMs. First, we provide a novel definition of durable opacity
  extending opacity to handle crashes and recovery in the context of
  NVM. Second, we develop a durably opaque version of an existing STM
  algorithm, namely the Transactional Mutex Lock (TML). Third, we
  design a proof technique for durable opacity based on refinement
  between TML and an operational characterisation of durable opacity
  by adapting the TMS2 specification. Finally, we apply this proof
  technique to show that the durable version of TML is indeed durably
  opaque. The correctness proof is mechanized within Isabelle.
\end{abstract}


\section{Introduction}\label{sec:intro}
Recent technological advances indicate that future architectures will
employ some form of {\em non-volatile memory} (NVM) that retains its
contents after a system crash (e.g., power outage). NVM is intended to
be used as an intermediate layer between traditional \emph{volatile
  memory} (VM) and secondary storage, and has the potential to vastly
improve system speed and stability. Software that uses NVM has the
potential to be more robust; in case of a crash, a system state before
the crash may be recovered using contents from NVM, as opposed to
being restarted from secondary storage. However, because the same data
is stored in both a volatile and non-volatile manner, and because NVM
is updated at a slower rate than VM, recovery to a consistent state
may not always be possible. This is particularly true for concurrent
systems, where coping with NVM requires introduction of additional
synchronisation instructions into a program.

This observation has led to the design of the first {\em persistent}
concurrent programming abstractions, so far mainly concurrent data
structures. Together with these, the associated notion of correctness,
i.e., linearizability~\cite{HeWi90}, has been transferred to NVM. This
resulted in the novel concept of {\em durable
  linearizability}~\cite{DBLP:conf/wdag/IzraelevitzMS16}.  A first
proof technique for showing durable linearizability of concurrent data
structures has been proposed by Derrick et
al.~\cite{DBLP:conf/fm/DerrickDDSW19}.

Besides concurrent data structures, software transactional memory
(STM) is the most important synchronization mechanism supporting
concurrent access to shared state. STMs provide an \emph{illusion of
  atomicity} in concurrent programs.  The analogy of STM is with
database transactions, which perform a series of accesses/updates to
shared data (via read and write operations) atomically in an
all-or-nothing manner. Similarly with an STM, if a transaction
commits, all its operations succeed, and in the aborting case, all its
operations fail.  The now (mainly) agreed upon correctness criterion
for STMs is {\em opacity}~\cite{2010Guerraoui}.  Opacity requires all
transactions (including aborting ones) to agree on a single sequential
history of committed transactions and the outcome of transactions has
to coincide with this history.

\smallskip\noindent
In this paper, we transfer STM and opacity to the novel field 
of non-volatile memory. This entails a number of steps. 
First, the correctness criterion of opacity has to be adapted 
to cope with crashes in system executions. 
Second, STM algorithms have to be extended to deal with the 
coexistence of volatile and non-volatile memory during execution 
and need to be equipped with recovery operations. 
Third, proof techniques for opacity need to be re-investigated 
to make them usable for durable opacity. 
In this paper, we provide contributions to all three steps. 

\begin{itemize} \setlength{\topsep}{0pt}
\item For the first step, we define {\em durable opacity} out of
  opacity in the same way that durable linearizability has been
  defined based on linearizability. Durable opacity requires the
  executions of STMs to be opaque even if they are interspersed with
  crashes.  This guarantees that the shared state remains consistent.
\item We exemplify the second step by extending the Transactional
  Mutex Lock (TML) of Dalessandro et al.~\cite{DalessandroDSSS10} to
  durable TML (\DTML).  To this end, TML needs to be equipped with a
  recovery operation and special statements to guarantee consistency
  in case of crashes. We do so by extending TML with a logging
  mechanism which allows to flush written, but volatile values to NVM
  during recovery.
\item For the third step, we build on a proof technique for opacity
  based on refinement between IO-automata.  This technique uses the
  automaton TMS2 \cite{DGLM13} which has been shown to implement
  opacity \cite{LLM12} using the PVS interactive theorem prover.  This
  automaton gives us a formal specification, which can be used as the
  abstract level in a proof of refinement. Furthermore, the
  IO-automaton framework is part of the standard Isabelle
  distribution.  For use as a proof technique for durable opacity,
  TMS2 is extended with a crash operation (mimicing system crashes and
  their effect on memory) to yield \DTMS. The automaton \DTMS is then
  proven to only have durably opaque executions.  Thereby we obtain an
  operational characterisation of durable opacity.
\end{itemize}

Finally, we bring all three steps together and apply our proof technique to show that durable TML is indeed durably opaque. 
This proof has been mechanized in the interactive prover Isabelle \cite{NipkowPW02}. 
Our mechanized proof
proceeds by encoding \DTMS and \DTML as IO-automata within Isabelle,
then proving the existence of a forward simulation, which in turn has
been shown to ensure trace refinement of IO-automata~\cite{LynchVaan95},  and hence guarantees durable opacity of \DTML.




\section{Transactional Memory and Opacity}
\label{sec:softw-trans-memory}

Software Transactional Memory (STM) provides programmers with an
easy-to-use synchronisation mechanism for concurrent access to shared
data, whereby blocks of code may be treated as transactions that
execute with an illusion of atomicity.
STMs usually provide a number
of operations to programmers: operations to start (\tmbegin) and
commit a transaction (\tmend), and operations to read and write shared
data (\tmread, \tmwrite). These operations can be called (invoked)
from within a client program (possibly with some arguments, e.g., the
variable to be read) and then will return with a response. Except for
operations that start transactions, all other operations might
potentially respond with \texttt{abort}, thereby aborting the whole
transaction. 

A widely accepted correctness condition for STMs that encapsulates
transactional phenomena is
\emph{opacity}~\cite{2010Guerraoui,GuerraouiK08}, which requires all
transactions, including aborting transactions to agree on a single
sequential global ordering of transactions. Moreover, no transactional
read returns a value that is inconsistent with the global ordering.

\subsection{Histories}
\label{sec:histories}
As standard in the literature, opacity is defined on the
\emph{histories} of an implementation.  Histories are sequences of
{\em events} that record all interactions between the implementation
and its clients.  An event is either an invocation ({\em inv}) or a
response ({\em res}) of a transactional operation.  For the TML
implementation, possible invocation and \emph{matching} response
events are given in Table~\ref{fig:stmevents}, where we assume $T$ is
a set of transaction identifiers, $L$ a set of addresses (or
locations) mapped to values from a set $V$.

The type $Mem \sdef L \rightarrow V$ describes the possible states of
the shared memory. We assume that initially all addresses hold the
value $0 \in V$.

\begin{table}[t]
\begin{tabular}{l@{\qquad\qquad}l}
  invocations~ & possible matching responses \\
  \hline 
  $inv_t(\tmbegin)$ & $res_t(\tmbegino)$ \\[1pt]
  $inv_t(\tmend)$ & $res_t(\tmendc)$, $res_t(\tmenda)$ \\[1pt]
  $inv_t(\tmread(x))$ & $res_t(\tmread(v))$, $res_t(\tmreada)$ \\[1pt]
  $inv_t(\tmwrite(x,v))$ & $res_t(\tmwriteo)$, $res_t(\tmwritea)$ 
\end{tabular} 
\caption{Events appearing in the histories of TML, where $t \in T$ is
  a transaction identifier, $x \in L$ is a location, and $v\in V$ a
  value}
\label{fig:stmevents}
\end{table}


 
We use the following notation on histories: for a history $h$,
$h \zproject t$ is the projection onto the events of transaction $t$
only and $h[i..j]$ the subsequence of $h$ from $h(i)$ to $h(j)$
inclusive. For a response event $e$, we let $rval(e)$ denote the value
returned by $e$; for instance $rval(\tmbegino) = \texttt{ok}$. If $e$
is not a response event, then we let $rval(e) = \bot$.

We are interested in three different types of histories
\cite{DBLP:conf/forte/ArmstrongDD17}. At the concrete level the TML
implementation produces histories where the events are
interleaved. 
At the
abstract level we are interested in {\em sequential histories}, which
are ones where there is no interleaving at any level - transactions
are atomic: a transaction completes before the next transaction
starts.  As part of the proof of opacity we use an intermediate
specification which has {\em alternating histories}, in which
transactions may be interleaved but operations (e.g., reads, writes)
are not interleaved.

A history $h$ is \emph{alternating} if $h = \epsilon$ or is an
alternating sequence of invocation and matching response events
starting with an invocation. For the rest of this paper, we assume
each process invokes at most one operation at a time, and hence,
assume that $h\zproject t$ is alternating for any history $h$ and
transaction $t$. Note that this does not necessarily mean $h$ is
alternating itself.  Opacity is defined for well-formed histories,
which formalises the allowable interaction between an STM
implementation and its clients. For every $t$,
$h\zproject t = \lseq s_0, \ldots, s_m\rseq$ of a well-formed history
is an alternating history such that $s_0 = inv_t(\tmbegin)$, for all
$0 < i < m$, event $s_i \neq inv_t(\tmbegin)$ and
$rval(s_i) \notin \{\texttt{commit}, \texttt{abort}\}$. Note that by
definition, well-formedness disallows transaction identifiers from
being reused. 
We say $t$ is {\em committed} if $rval(s_m) = \texttt{commit}$ and
{\em aborted} if $rval(s_m) = \texttt{abort}$.  In these cases, the
transaction $t$ is {\em completed}, otherwise $t$ is {\em live}.  A
history is {\em well-formed} if it consists of transactions only and
there is at most one live transaction per process.




\subsection{Opacity}
\label{sec:opacity}

Opacity \cite{2010Guerraoui,GuerraouiK08} compares concurrent
histories generated by an STM implementation to sequential histories
and can be seen as a strengthening of serializability to accommodate
aborted transactions. Below, we first formalise the sequential history
semantics, then consider opaque histories.

\paragraph{Sequential history semantics.}


\noindent 
A sequential history has to ensure that the behaviour is meaningful
with respect to the reads and writes of the transactions. 

\begin{definition}[Valid history]\label{def:valid-hist}
  Let $h = ev_0, \ldots ,
  ev_{2n-1} $ be a sequence of alternating invocation and matching response
  events 
  starting with an invocation and
  ending with a response.

  We say $h$ is {\em valid} if there exists a sequence of states
  $\sigma_0, \ldots , \sigma_{n}$ such that $\sigma_0(l) = 0$ for all
  $l \in L$, and for all $i$ such that $0 \leq i < n$ and $t \in T$:
  \begin{enumerate}
  \item if $ ev_{2i} = inv_t(\tmwrite(l,v))$ and
    $ev_{2i+1} = res_t(\tmwriteo)$ 
    then $\sigma_{i+1} = \sigma_{i}[l:=v]$, 
    
  \item if $ev_{2i} = inv_t(\tmread(l))$ and
    $ev_{2i + 1} = res_t(\tmread(v))$ then $\sigma_i(l)=v$ and
    $\sigma_{i+1} = \sigma_{i}$, 
  \item for all other pairs of events (reads and writes with an abort
    response, as well as begin and commit events) we require
    $\sigma_{i+1} = \sigma_{i}$.
  \end{enumerate}
  We write $\llbracket h \rrbracket(\sigma)$ if $\sigma$ 
  is a sequence of states that makes $h$ valid (since
  the sequence is unique, if it exists, it can be viewed
  as the semantics of $h$).
\end{definition}

The point of TM is that the effect of the writes only takes place if
the transaction commits. Writes in a transaction that abort do not
affect the memory.  However, all reads, including those executed by
aborted transactions, must be consistent with previously committed
writes. Therefore, only some histories of an object reflect ones that
could be produced by a TM. We call these the {\em legal} histories,
and they are defined as follows.

\begin{definition}[Legal histories]
  \label{legal}
  Let $hs$ be a non-interleaved history and $i$ an index of $hs$.  
  Let $hs'$ be the projection of $hs[0..(i-1)]$ onto all events of committed
  transactions plus the events of the transaction to which $hs(i)$
  belongs. 
  Then we say \emph{$hs$ is legal at $i$} whenever $hs'$ is valid.
  We say \emph{$hs$ is legal} iff it is legal at each index
  $i$. 
\end{definition}

\noindent This allows us to define sequentiality for a single history,
which we lift to the level of specifications.  

\begin{definition}[Sequential history]
  \label{def:sequential}  
  A well-formed history $hs$ is {\em sequential} if it is
  non-interleaved and legal. 
  We denote by ${\cal S}$  the
  set of all possible well-formed sequential histories. 
\end{definition}

\paragraph{Opaque histories.}
Opacity is defined by matching a concurrent history to a sequential
history such that (a) both histories consist of the same events, and
(b) the real-time order of transactions is preserved.
%
For (b), the {\em real-time order} on transactions $t_1$ and $t_2$ in
a history $h$ is defined as $t_1 \prec_h t_2$ if $t_1$ is a completed
transaction and the last event of $t_1$ in $h$ occurs before the first
event of $t_2$.

A given concrete history may be incomplete, i.e., it may contain pending operations,
represented by invocations that do not have matching responses. Some of these pending operations
may have taken effect, and others may not. The corresponding sequential history
however must decide: either by adding a suitable matching response
event for the pending invocation (the effect has taken place), or by
removing the pending invocation (no effect yet). Therefore, we define a
function $complete(h)$ that constructs all possible completions of $h$
by appending matching responses for some pending invocations and removing
all the other pending invocations. This is similar to the treatment of completions in
the formalisation of linearizability \cite{HeWi90}. The sequential
history then must have the same events as those of one of the results
returned by $complete(h)$.


\begin{definition}[Opaque history]
  \label{def:opaque}  
  \label{opaquedef}
  A history $h$ is {\em end-to-end opaque} iff for some
  $hc \in complete(h)$, there exists a sequential history
  $hs \in {\cal S}$ such that for all $t \in T$,
  $hc \zproject t = hs \zproject t$ and
  $\prec_{hc} \subseteq \prec_{hs}$.  A history $h$ is {\em opaque}
  iff each prefix $h'$ of $h$ is end-to-end opaque; a set of histories
  ${\cal H}$ is {\em opaque} iff each $h \in {\cal H}$ is opaque; and
  an STM implementation is \emph{opaque} iff its set of histories is
  opaque.
\end{definition}


\section{STMs over Persistent Memory}

We now consider STMs over a non-volatile memory model comprising two
layers: a {\em volatile store}, whose contents are wiped clean when a
system crashes (e.g., due to power loss), and a {\em persistent
  store}, whose state is preserved after a crash and available for use
upon reboot. During normal program execution, contents of the volatile
store may be transferred to the persistent store by the system.  The
main idea behind programs for this memory model is to include a
\emph{recovery procedure} that executes over the persistent store and
resets the system into a consistent (safe) state. To achieve this, a
programmer can control transfer of information from volatile to
persistent store using a {\tt FLUSH(a)} operation, ensuring that the
information in address {\tt a} is saved in the persistent store.

For STMs, we introduce a new notion of consistency: \emph{durable
  opacity} which we define in \refsec{sec:durable-opacity}. Durable
opacity extends opacity~\cite{2010Guerraoui,GuerraouiK08} in exactly
the same way that durable
linearizability~\cite{DBLP:conf/wdag/IzraelevitzMS16} extends
linearizability~\cite{HeWi90}, namely a history that contains crashes
is durably opaque precisely when the same history with crashes removed
is opaque. We present an example STM implementation that satisfies
durable opacity in \refsec{sec:exampl-trans-mutex}, extending
Dalessandro et al.'s Transactional Mutex
Lock~\cite{DBLP:conf/ppopp/DalessandroSS10}.

\subsection{Durable Opacity}
\label{sec:durable-opacity}
Durable opacity is a correctness condition that is defined over
\emph{histories} that record the \emph{invocation} and \emph{response}
events of operations executed on the transactional memory like
opacity.  Unlike opacity, durably opaque histories record system crash
events, thus may take the form:
$H = h_0 \crash_1 h_1 \crash_2 ...h_{n-1} \crash_n h_n$, where each
$h_i$ is a history (containing no crash events) and $\crash_i$ is the
$i$th crash event.  Following Izraelevitz et
al.\cite{DBLP:conf/wdag/IzraelevitzMS16}, for a history $h$, we let
$\ops(h)$ denote $h$ restricted to non-crash events, thus for $H$
above, $ops(H) = h_0 h_1 \dots h_{n-1} h_n$, which contains no crash
events. We call the subhistory $h_i$ the {\em $i$-th era} of
$h$.  

The definition of a well-formed history is now updated to include
crash events. A history is \emph{durably well-formed} iff $ops(h)$ is
well formed and every transaction identifier appears in at most one
era. Thus, we assume that when a crash occurs, all running
transactions are aborted.  

\begin{definition}[Durably opaque history]
  \label{def:duropaque}
  A history $h$ is {\em durably opaque} iff it is durably well-formed
  and $ops(h)$ is opaque.
\end{definition}

\subsection{Example: Durable Transactional Mutex Lock}
\label{sec:exampl-trans-mutex}

We now develop a durably opaque STM: a persistent memory version of
the Transactional Mutex Lock (TML)~\cite{DalessandroDSSS10}, as given
in Fig.~\ref{fig:DTML}.  TML adopts a strict policy for transactional
synchronisation: as soon as one transaction has successfully written
to a variable, all other transactions running concurrently will be
aborted when they invoke another read or write operation. To enforce
this policy, TML uses a global counter {\tt glb} (initially $0$) and
local variable {\tt loc}, which is used to store a copy of {\tt
  glb}. Variable {\tt glb} records whether there is a {\em live
  writing transaction}, i.e., a transaction that has started, has not
yet ended nor aborted, and has executed (or is executing) a write
operation. More precisely, {\tt glb} is odd if there is a live writing
transaction, and even otherwise. Initially, we have no live writing
transaction and thus {\tt glb} is 0 (and hence even).

A second distinguishing feature of TML is that it performs writes in
an \emph{eager} manner, i.e., it updates shared memory during the
write operation\footnote{This is in contrast to lazy implementations
  that defer transactional writes until the commit operation is
  executed (e.g.,
  \mbox{\cite{DBLP:conf/wdag/DiceSS06,DBLP:conf/ppopp/DalessandroSS10}}).}.
This is potentially problematic in a persistent memory context since
writes that have completed may not be committed if a crash occurs
prior to executing the commit operation. That is, writes of
uncommitted transactions should not be seen by any transactions that
start after a crash occurs. Our implementation makes use of an {\em
  undo log} mapping addresses to their persistent memory values prior
to executing the first write operation for that address. Logged values
are made persistent before the address is overwritten. Thus, if a
crash occurs prior to a transaction committing, it is possible to
recover the transaction to a safe state by undoing uncommitted
transactional writes.

Operation \verb+TMBegin+ copies the value of {\tt glb} into its local
variable {\tt loc} and checks whether {\tt glb} is even. If so, the
transaction is started, and otherwise, the process attempts to start
again by rereading {\tt glb}. A \verb+TMRead+ operation succeeds as
long as {\tt glb} equals {\tt loc} (meaning no writes have occurred
since the transaction began), otherwise it aborts the current
transaction. The first execution of \verb+TMWrite+ attempts to
increment {\tt glb} using a {\tt cas} (compare-and-swap), which
atomically compares the first and second parameters, and sets the
first parameter to the third if the comparison succeeds. If the {\tt
  cas} attempt fails, a write by another transaction must have
occured, and hence, the current transaction aborts. Otherwise {\tt
  loc} is incremented (making its value odd) and the write is
performed. Note that because {\tt loc} becomes odd after the first
successful write, all successive writes that are part of the same
transaction will perform the write directly after testing {\tt loc} at
line $W1$. Further note that if the {\tt cas} succeeds, {\tt glb}
becomes odd, which prevents other transactions from starting, and
causes all concurrent live transactions still wanting to read or write
to abort. Thus a writing transaction that successfully updates {\tt
  glb} effectively locks shared memory. Operation \verb+TMEnd+ checks
to see if a write has occurred by testing whether {\tt loc} is odd. If
the test succeeds, {\tt glb} is set to {\tt loc+1}. At line {\tt E2},
{\tt loc} is guaranteed to be equal to {\tt glb}, and therefore this
update has the effect of incrementing {\tt glb} to an even value,
allowing other transactions to begin.

Our implementation uses a durably linearizable
\cite{DBLP:conf/fm/DerrickDDSW19,DBLP:conf/wdag/IzraelevitzMS16} set
or map data structure {\tt log}, such as the one described by Zuriel
et al.~\cite{DBLP:journals/pacmpl/ZurielFSCP19}. A durably
linearizable operation is guaranteed to take effect in persistent
memory prior to the operation returning. In Fig.~\ref{fig:DTML}, we
use use operations {\tt pinsert()}, {\tt pempty()} and {\tt pdelete()}
to stress that these operations are durably linearizable.

Our durable TML algorithm ({\sc dTML}) makes the following adaptations
to TML. Note the the operations build on a model of a crash that
resets volatile memory to persistent memory.
\begin{itemize} \setlength{\topsep}{0pt}
\item Within a write operation writing to address {\tt addr}, prior to
  modifying the value at {\tt addr}, we record the existing
  address-value pair in {\tt log}, provided that {\tt addr} does not
  already appear in the undo log (lines {\tt W4} and {\tt W5}). After
  updating the value (which updates the value of {\tt addr} in the
  volatile store), the update is flushed to persistent memory prior to
  the write operation returning (line {\tt W7}).
\item We introduce a recovery operation that checks for a non-empty
  log and transfers the logged values to persistent memory, undoing
  any writes that have completed (but not committed) before the crash
  occurred. Since a crash could occur during recovery, we transfer
  values from the undo log to persistent memory one at a time.
\item In the commit operation, we note that we distinguish a
  committing transaction as one with an odd value for {\tt loc}. For a
  writing transaction, the log must be cleared by setting it to the
  empty log (line {\tt E2}). Note that this is the point at which a
  writing transaction has definitely committed since any subsequent
  crash and recovery would no longer undo the writes of this
  transaction.
\end{itemize}


\begin{figure}[t]
\small
\noindent
\begin{minipage}[t]{0.45\columnwidth}
\begin{lstlisting}
TMBegin:                         
B1 do loc := glb                  
B2 until even(loc); 
   return ok;                    
                                         
TMRead(addr):                    
R1 val := *addr;                 
R2 if (glb = loc) then
     return val; 
   else return abort; 
       
Recover(): 
C1 while $\neg$ log.isEmpty() 
C2   SOME (addr, val).
        (addr, val) $\in$ log; 
C3   *addr := val ; 
C4   FLUSH(addr) ; 
C5   log.pdelete((addr, val));
C6 glb := 0
\end{lstlisting} 
\end{minipage}
\hfill 
\begin{minipage}[t]{0.52\columnwidth}
\begin{lstlisting}[escapeinside={(*}{*)}]
TMCommit:   
E1 if odd(loc) then 
E2   log.pempty(); 
E3   glb := loc + 1;  
   return commit; 

TMWrite(addr,val): 
W1 if even(loc) then 
W2   if $\neg$ cas(glb,loc,loc+1) then
        return abort;
W3   else loc++;
W4 if $\all$ v. (addr, v) $\notin$  log then
W5    log.pinsert((addr, *addr)); 
W6 *addr := val; 
W7 FLUSH(addr); 
   return ok;
\end{lstlisting} 
\end{minipage}
\caption{A durable Transactional Mutex Lock (\DTML). Initially: {\tt
    glb = 0, log = emptyLog()}. Line numbers for {\tt return}
  statements are omitted and we use {\tt *addr} for the value of {\tt
    addr} }
\label{fig:DTML}
\end{figure}

\section{Proving Durable Opacity}
\label{sec:prov-opac-meth}

Previous
works~\cite{DBLP:conf/forte/ArmstrongDD17,DBLP:conf/opodis/DohertyDDSW16,DBLP:journals/csur/DongolD15,DBLP:conf/forte/ArmstrongD17}
have considered proofs of opacity using the operational TMS2
specification~\cite{DGLM13}, which has been shown to guarantee opacity
\cite{LLM12}. The proofs show refinement of the implementation against
the TMS2 specification using either forward or backward
simulation. For durable opacity, we use a similar proof strategy. In
\refsec{sect:TMS2-def}, we develop the \DTMS operational
specification, a durable version of the TMS2 specification, that we
prove satisfies durable opacity. Then, in
\refsec{sec:verify-durable-opac}, we establish a simulation between
\DTML and \DTMS.

\subsection{Background: IOA, Refinement and Simulation}
We use Input/Output Automata (IOA) \cite{LT87} to model both the
implementation, \DTML, and the specification, \DTMS.
\begin{definition}
  An \emph{Input/Output Automaton (IOA)} is a labeled transition
  system $A$ with a set of \emph{states} $states(A)$, a set of
  \emph{actions} $acts(A)$, a set of \emph{start states}
  $start(A)\subseteq states(A)$, and a \emph{transition relation}
  $trans(A)\subseteq states(A)\times acts(A)\times states(A)$ (so that
  the actions label the transitions).
\end{definition}
The set $acts(A)$ is partitioned into input actions $input(A)$, output
actions $output(A)$ and internal actions $internal(A)$.  The internal
actions represent events of the system that are not visible to the
external environment. The input and output actions are externally
visible, representing the automaton's interactions with its
environment. Thus, we define the set of {\em external actions},
$external(A) = input(A)\cup output(A)$. We write
$s \stackrel{a}{\longrightarrow}_A s'$ iff $(s, a, s') \in trans(A)$.

An {\em execution} of an IOA $A$ is a
sequence $\sigma = s_0 a_0 s_1 a_1 s_2 \dots s_n a_n s_{n+1}$ of
alternating states and actions, 
such that $s_0 \in start(A)$ and 
for all states $s_i$,
$s_i \stackrel{a_{i}}{\longrightarrow}_A s_{i+1}$. A {\em reachable}
state of $A$ is a state appearing in an execution of $A$. An {\em
  invariant} of $A$ is any superset of the reachable states of $A$
(equivalently, any predicate satisfied by all reachable states of
$A$). A {\em trace} of $A$ is any sequence of (external) actions
obtained by projecting the external actions of any execution of
$A$. The set of traces of $A$, denoted $traces(A)$, represents $A$'s
externally visible behaviour.

For automata $C$ and $A$, we say that $C$ is a {\em refinement} of $A$
iff $traces(C) \subseteq traces(A)$.  We show that $C$ is a refinement
of $A$ by proving the existence of a {\em forward simulation}, which
enables one to check step correspondence between the transitions of
$C$ and those of $A$. The definition of forward simulation we use is
adapted from that of Lynch and Vaandrager \cite{LynchVaan95}.
\begin{definition}
\label{def:for-sim}
A \emph{forward simulation} from a concrete IOA $C$ to an abstract IOA
$A$ is a relation $R \subseteq states(C) \times states(A)$ such that 
each of the following holds.  \smallskip

\noindent \emph{Initialisation}. $\forall cs \in start(C).\ \exists as \in start(A).\ R(cs, as)$\smallskip


\noindent \emph{External step correspondence}.

\hfill $\begin{array}[t]{@{}l@{}}
  \forall cs \in reach(C), as \in reach(A), a \in external(C), cs' \in states(C).\ \\
  \qquad R(cs, as) \wedge cs \stackrel{a}{\longrightarrow}_C cs' \imp \exists as' \in states(A).\ R(cs', as') \wedge as \stackrel{a}{\longrightarrow}_A as'
\end{array}
$ \hfill {}\smallskip



\noindent \emph{Internal step correspondence}.

\hfill
$\begin{array}[t]{@{}l@{}}
   \forall cs \in reach(C), as \in reach(A), a \in internal(C), cs' \in states(C).\ \\
   \qquad R(cs, as) \wedge cs \stackrel{a}{\longrightarrow}_C cs' \imp \\
   \qquad \begin{array}[t]{@{}l@{}}
     R(cs', as) \lor  
     \exists a' \in internal(A), as' \in states(A).\ R(cs', as') \wedge as \stackrel{a'}{\longrightarrow}_A as'
   \end{array}
 \end{array}
 $\hfill{}


  
\end{definition}
Forward simulation is {\em sound} in the sense that if there is a
forward simulation between $A$ and $C$, then $C$ refines $A$
\cite{LynchVaan95,MullerIOA1998}.

\subsection{IOA for \DTML}
\label{sect:TML-def}

We now provide the IOA model of \DTML. The state of \DTML
(\reffig{fig:DTML}) comprises global (shared) variables $glb \in \nat$
(modelling {\tt glb} in volatile memory); $log \in L \pfun V$, where
$\pfun$ denotes a partial function (modelling {\tt log} in persistent
memory); the volatile memory store $vstore \in L \to V$; and the
persistent memory store $pstore \in L \to V$. We also use the
following transaction-local variables: the program counter
$pc \in T \to PC$, $loc \in T \fun \nat$, the input address
$addr \in T \fun V$, the input value $val \in T \fun V$. We also make
use of an auxiliary variable $writer$ whose value is either the
transaction id of the current writing transaction (if one exists), or
$None$ (if no writing transaction is currently running).

Execution of the program is modelled by defining an IOA transition for
each atomic step of \reffig{fig:DTML}, using the values of $pc_t$ (for
transaction $t$) to model control flow. Each action that starts a new
operation or returns from a completed operation is an external
action. The crash action is also external. All other actions
(including flush and recovery) are internal actions.

To model system behaviours (crash, system flush and recovery), we
reserve a special transaction id $syst$. A crash and system flush is
always enabled, and hence can always be selected for
execution. Recovery steps are enabled after a crash has taken place
and are only executed by $syst$.  The effect of a flush is to copy the
value of the address being flushed from $vstore$ to $pstore$. Note
that a flush can also be executed at specific program locations. In
\DTML, a flush of {\tt addr} occurs at lines {\tt W7} and {\tt
  C5}. The effect of a crash is to perform the following:
\begin{itemize}
\item set the
volatile store to the persistent store (since the volatile store is lost),  
\item set the program counters of all \emph{in-flight transactions}
  (i.e., transactions that have started but not yet completed) to
  $aborted$ to ensure that these transaction identifiers are not
  reused after the system is rebooted, and
\item set the status of $syst$ to {\tt C1} to model that a recovery is now
  in progress.
\end{itemize}
In our model, it is possible for a system to crash during
recovery. However, no new transaction may start until after the
recovery process has completed.

\subsection{IOA for \DTMS}
\label{sect:TMS2-def}

In this section, we describe the \DTMS specification, an operational
model that ensures durable opacity, which is based on
TMS2~\cite{DGLM13}. TMS2 itself has been shown to strictly imply
opacity \cite{LLM12}, and hence has been widely used as an
intermediate specification in the verification of transactional memory
implementations~\cite{DBLP:conf/forte/ArmstrongDD17,DDSTW15,DBLP:conf/forte/ArmstrongD17,DBLP:conf/opodis/DohertyDDSW16}.

We let $f \oplus g$ denote functional override of $f$ by $g$, e.g.,
$f \oplus \{x \mapsto u, y \mapsto v\} = \lambda k.\ {\bf if}\ k = x\
{\bf then}\ u \ {\bf else if}\ k = y \ {\bf then}\ v \ {\bf else}\
f(k)$.

\begin{figure}
{\small
{\bf State variables:}\\
$mems : seq (L \to V)$, initially satisfying $\dom mems = \{0\}$ and $initMem(mems(0))$\\
$pc_t : PCVal$, for each $t\in T$, initially $pc_t=\pcNotStarted$ for all $t\in T$\\
$beginIdx_t : \nat$ for each $t\in T$, unconstrained initially\\
$rdSet_t: L \pfun V$, initially empty for all $t\in T$\\
$wrSet_t: L \pfun V$, initially empty for all $t\in T$\\
\\
{\bf Transition relation:} \\[1em]
\bgroup
\setlength{\tabcolsep}{1em}
\newlength{\myrowsep}
\setlength{\myrowsep}{4em} 
\newlength{\myrowsepa}
\setlength{\myrowsepa}{3em} 
\newlength{\myrowsepb}
\setlength{\myrowsepb}{5.5em} 
\begin{tabular}{@{}l@{\qquad\quad}l@{}}
\action{\beginInv{t}}
{pc_t = \pcNotStarted}
{pc_t := \pcBeginPending\\&
 beginIdx_t := len(mems) - 1}
&
\action{\beginResp{t}}
{pc_t = \pcBeginPending}
{pc_t := \pcReady}
\\[\myrowsep]
\action{\readInv{t}{l}}
{pc_t = \pcReady}
{pc_t := \pcDoRead(l)}
&
\action{\readResp{t}{v}}
{pc_t = \pcReadResp(v)}
{pc_t := \pcReady}
\\[\myrowsepa]
\action{\writeInv{t}{l, v}}
{pc_t = \pcReady}
{pc_t := \pcDoWrite(l, v)}
&
\action{\writeResp{t}}
{pc_t = \pcWriteResp}
{pc_t := \pcReady}
\\[\myrowsepa]
\action{\commitInv{t}}
{pc_t = \pcReady}
{pc_t := \pcDoCommit}
&
\action{\commitResp{t}}
{pc_t = \pcCommitResp}
{pc_t := \pcCommitted}
\\[\myrowsepa] 
\action{\abortResp{t}}
{pc_t \notin \{\pcNotStarted, \pcReady, \\
  & \pcCommitResp, \pcCommitted, \pcAborted\}}
{pc_t := \pcAborted}
  &
    \action{\doWrite{t}{l, v}}
    {pc_t = \pcDoWrite(l, v)}
    {pc_t := \pcWriteResp\\&
  wrSet_t := wrSet_t\oplus \{l \to v\}
  }
  \\[\myrowsep]
\action{\doCommitReadOnly{t}{n}}
{pc_t = \pcDoCommit\\&
 \dom(wrSet_t) = \emptyset\\&
 validIdx(t, n)}
{pc_t := \pcCommitResp}
&
\action{\doCommitWriter{t}}
{pc_t = \pcDoCommit\\&
 rdSet_t \subseteq last(mems)}
{pc_t := \pcCommitResp\\&
 mems := mems \cat (last(mems) \oplus wrSet_t)} 
\\[\myrowsepb]
\action{\doRead{t}{l, n}}
{pc_t = \pcDoRead(l)\\&
 l\in \dom(wrSet_t) \vee validIdx(t, n)}
{\sif\ l \in \dom(wrSet_t)\ \sthen\\&
 \ \ \ pc_t := \pcReadResp(wrSet_t(l))\\&
 \selse\ v := mems(n)(l)\\&
 \ \ \ \ \ \ \ \  pc_t := \pcReadResp(v)\\&
 \ \ \ \ \ \ \ \ rdSet_t := rdSet_t\oplus \{l \to v\}}
&
\begin{tabular}[t]{@{}l@{}}
  \action{\crashed_{t}}
  {{\it t = syst}}
  {pc := \lambda t: \Tids. \\
  & \qquad \sif\ t \neq syst \wedge {} \\
  & \qquad \quad pc_t \notin \{notStarted, committed\} \\
  & \qquad \sthen\ aborted \\
  & mems = \langle last(mems) \rangle}
\end{tabular}
\end{tabular}
\egroup
\smallskip

{\bf where}
$\begin{array}[t]{r@{~~}c@{~~}l}
  validIdx(t, n) &\sdef & beginIdx_t \leq n < len(mems) \wedge
                          rdSet_t \subseteq mems(n) 
\end{array}
$
}
\caption{The state space and transition relation of \DTMS, which
  extends TMS2 with a crash event} \label{fig:tms2}
\end{figure}

Formally, \DTMS is specified by the IOA in \reffig{fig:tms2}, which
describes the required ordering constraints, memory semantics and
prefix properties. We assume a set $L$ of locations and a set $V$ of
values.  Thus, a memory is modelled by a function of type $L \to
V$.   A key feature of \DTMS (like TMS2) is that it keeps track of a
{\em sequence} of memory states, one for each committed writing
transaction. This makes it simpler to determine whether reads are
consistent with previously committed write operations. Each committing
transaction containing at least one write adds a new memory version to
the end of the memory sequence. However, unlike TMS2, following
\cite{DBLP:conf/fm/DerrickDDSW19}, the memory state is considered to
be the persistent memory state. Interestingly, the volatile memory
state need not be modelled.

The state space of \DTMS has several components. The first, $mems$ is
the sequence of {\em memory} states. For each transaction $t$ there is
a program counter variable $pc_t$, which ranges over a set of {\em
  program counter values}, which are used to ensure that each
transaction is well-formed, and to ensure that each transactional
operation takes effect between its invocation and response. There is
also a a {\em begin index} variable $beginIdx_t$, that is set to the
index of the most recent memory version when the transaction
begins. This variable is critical to ensuring the real-time ordering
property between transactions.  Finally, there is a {\em read set},
$rdSet_t$, and a {\em write set}, $wrSet_t$, which record the values
that the transaction has read and written during its execution,
respectively. 

The read set is used to determine whether the values that have been
read by the transaction are consistent with the same version of memory
(using $validIdx$). The write set, on the other hand, is required
because writes in \DTMS are modelled using {\em deferred update}
semantics: writes are recorded in the transaction's write set, but are
not published to any shared state until the transaction commits.


The $crash$ action models both a crash and a recovery. We require that
it is executed by the system thread $syst$. It sets the program
counter of every in-flight transaction to $aborted$, which prevents
these transactions from performing any further actions in the era
following the crash (for the generated history). Note that since
transaction identifiers are not reused, the program counters of
completed transactions need not be set to any special value (e.g.,
$crashed$) as with durable
linearizability~\cite{DBLP:conf/fm/DerrickDDSW19}.  Moreover, after
restarting, it must not be possible for any new transaction to
interact with memory states prior to the crash. We therefore reset the
memory sequence to be a singleton sequence containing the last memory
state prior to the crash.

The following theorem ensures that \DTMS can be used as an
intermediate specification in our proof method. 

\begin{theorem}
\label{theor:dtsm-sound}
Each trace of\, \DTMS is durably opaque.
\end{theorem}

\begin{proof}[Sketch] First we recall that $\TMS$ is exactly the same
  as the automaton in \reffig{fig:tms2}, but without a crash
  operation. The proof proceeds by showing that for any history
  $h \in traces(\DTMS)$, we have that $ops(h) \in traces(\TMS)$. Then
  since $ops(h)$ is opaque, we have that $h$ is durably opaque. We
  establish a formal relationship between $h$ and $ops(h)$ by
  establishing a \emph{weak simulation} between $\DTMS$ and $\TMS$
  such that $\{ops(h) | h \in traces(\DTMS)\} \subseteq
  traces(\TMS)$. The simulation is weak since the external
  \emph{crash} action in \DTMS has no matching counterpart in \TMS.

  The simulation relation we use captures the following.  Any
  transaction $t$ of $\DTMS$ that is aborted due to a crash will set
  $pc_t$ to $aborted$ without executing $\abortResp{t}$. This
  difference can easily be compensated by the simulation relation. A
  second difference is that $mems$ is reset to $last(mems)$ in $\DTMS$
  when a crash occurs, and hence there is a mismatch between $mems$ in
  $\DTMS$ and in $\TMS$. Let $ds$ be a state of \DTMS and $as$ a state
  of \TMS. To compensate for the difference between $ds.mems$ and
  $as.mems$, we introduce an auxiliary variable ``$allMems$'' to $ds$
  that records memories corresponding to all committed writing
  transactions in $\DTMS$. We have the property that $ds.mems$ of
  \DTMS is a suffix of $ds.allMems$ and that $ds.allMems = as.mems$.
\end{proof}


\section{Durable Opacity of\, \DTML}
\label{sec:verify-durable-opac}

We now describe the simulation relation used in the Isabelle
proof.\footnote{All Isabelle theory files related to this proof are
  provided as ancillary files to this submission.}

Our simulation relation is divided into two relations: a {\em global
  relation} $globalRel$, and a transactional relation $txnRel$. The
global relation describes how the shared states of the two automata are
related, and the transaction relation specifies the relationship
between the state of each transaction in the concrete automaton, and
that of the transaction in the abstract automaton.  The simulation
relation itself is then given by:
\begin{align*}
simRel(cs, as) = globalRel(cs, as) \wedge \forall t\in T \st txnRel(cs, as, t)
\end{align*}

We first describe $globalRel$, which assumes the following auxiliary
definitions where $cs$ is the concrete state (of \DTML) and $as$ is
the abstract state (of \DTMS). These definitions are used to
compensate for the fact that the commit of a writing transaction in
the \DTML algorithm takes effect (i.e., linearizes) at line {\tt E2}
when the log is set to empty.
\begin{align*}
  writes(cs, as) & = {\bf if}\ cs.writer = t \wedge
  pc_{t} \neq E3\ {\bf then}\ as.wrSet_t\ {\bf else}\
  \emptyset \\
  logical\_glb(cs) & = {\bf if}\ cs.writer = t \wedge
                     pc_{t} = E3\  {\bf then}\ cs.glb + 1 \ {\bf else}\ cs.glb \\
  write\_count(cs) & = \left\lfloor\frac{logical\_glb(cs)}{2}\right\rfloor
\end{align*}
Function $writes(cs, as)$ returns the (abstract) write set of the
writing transaction. This is the write set of the writing transaction,
$t$, in the abstract state $as$ provided $t$ hasn't already linearized
its commit operation, and is the empty set otherwise. Function
$logical\_glb(cs)$ compensates for a lagging value of $glb$ after a
writing transaction's commit operation is linearized. Namely, it
returns the $glb$ incremented by $1$ if a writer is already at
$E3$. Finally, $write\_count(cs)$ is used to determine the number of
committed writing transactions in $cs$ since the most recent crash
since $cs.glb$ is initially $0$ and reset to $0$ by the recovery
operation, and moreover, $cs.glb$ is incremented twice by each writing
transaction: once at line {\tt W2} and again at line {\tt E2} when the
writing transaction commits.

Relation $globalRel$ comprises three main parts.  We assume a program
counter value $RecIdle$ which is true for $pc_{syst}$ iff $syst$ is
not executing the recovery procedure. 
\begin{align}
  \omit \rlap{$globalRel(cs, as)$  = } \nonumber \\
  & (pc_{syst} = RecIdle \imp cs.vstore = (last (as.mems) \oplus writes (cs, as)) \wedge {} \label{grel1}\\
  & \qquad \qquad \qquad \qquad write\_count(cs) + 1 = length (as.mems))) \wedge {} \label{grel2} \\
  &  (cs.vstore \oplus cs.log) = last(mems(as))\wedge {} \label{grel3} \\
  & \forall t. t \neq syst \wedge cs.pc_t = NotStarted \imp as.pc_t = NotStarted  \label{grel4}
\end{align} 
Conditions \eqref{grel1} and \eqref{grel2} assume that a recovery
procedure is not in progress. By \eqref{grel1}, the concrete volatile
store is the last memory in $as.mems$ overwritten with the write set
of an in-flight writing transaction that has not linearized its commit
operation. By~\eqref{grel2}, the number of memories recorded in the
abstract state (since the last crash) is equal to
$write\_count(cs) + 1$. 
By \eqref{grel3}, the last abstract (persistent) store can by
calculated from $cs.vstore$ by overriding it with the mappings in
log. Note that this is equivalent to undoing all uncommitted
transactional writes. Finally, \eqref{grel4} ensures that every
identifier for a transaction that has not started at the concrete
level also has not started at the abstract level.

We turn now to $txnRel$. Its specification is very similar to the
specification of $txnRel$ in the proof of
TML~\cite{DBLP:journals/fac/DerrickDDSTW18}. Therefore, we only
provide a brief sketch below; an interested reader may
consult~\cite{DBLP:journals/fac/DerrickDDSTW18} for further
details. Part of $txnRel$ maps concrete program counters to their
abstract counterparts, which enables steps of the concrete program to
be matched with abstract steps. For example, concrete $pc$ values {\tt
  W1}, {\tt W2}, \dots, {\tt W6} correspond to abstract $pc$ value
$\pcDoWrite(cs.addr_t,cs.val_t)$, whereas {\tt W7} corresponds to
$\pcWriteResp$, indicating that, in our proof, execution of line {\tt
  W6} corresponds to the execution of an abstract
$\doWrite{t}{cs.addr_t,cs.val_t}$ operation. Moreover, as in the
proof of TML~\cite{DBLP:journals/fac/DerrickDDSTW18}, a set of
assertions are introduced to establish
$as.validIdx(t, write\_count(cs))$ for all in-flight
transactions $t$, which ensures that each transactional read and write
is valid with respect to some memory snapshot.




Relation $txnRel$ must also provide enough information to enable
linearization of a commit operation against the correct abstract
step. Note that \DTMS distinguishes between read-only and writing
transactions by checking emptiness of the write set of the committing
transaction. To handle this, we exploit the fact that in \DTML,
writing transactions have an odd $loc$ value if the {\tt cas} at line
{\tt W2} is successful and {\tt loc} is incremented at {\tt W3},
indicating that a writing transaction is in
progress.  

Finally, $txnRel$ must ensure that the recovery operation is such that the
volatile store matches the last abstract store in $mems$ prior to the
crash. To achieve this, we require that $length(as.mems) = 1$ when
$syst$ is executing the recovery procedure, and the volatile store for
the address being flushed at {\tt C3} matches the abstract state
before the crash, i.e.,
$cs.vstore(cs.addr_t) = ((as.mems)(0)) (cs.addr_t)$. Since the
recovery loop only terminates after the log is emptied, this ensures
that the concrete memory state is consistent with the abstract memory
prior to executing any transactions after a crash has occurred. 

In order to prove that our simulation relation is maintained
by each step of the algorithm, we must use
certain invariants of the \DTML model. These invariants are
similar to the corresponding invariants used in a proof of the
original TML algorithm for the conventional volatile RAM model
(see \cite{DBLP:journals/fac/DerrickDDSTW18} for details).
For example, our invariants imply that
there is at most one writing transaction, and there is no such transaction
when $glb$ is even. The main additional invariant that we use for 
\DTML constrains the possible differences between volatile and persistent
memory: volatile and persistent memory are identical except for any
location that has been written by a writer or by the recovery procedure
but not yet flushed. This simple invariant combined with the global relation
is enough to prove that the memory state after each crash is correct.
Our \DTML invariants have been verified in Isabelle, and can be
found in the Isabelle files.

\section{Related Work}
Although there is existing research on extending the definition of
linearizability to durable systems, there is comparatively less work
on extending other notions of transactional memory correctness such
as, but not limited to, opacity to durable systems.  Various systems
attempt to achieve atomic durability, transform general objects to
persistent objects and provide a secure interface of persistent
memory. The above goals usually require the use of logging which can
be software or hardware based. Raad et al have proposed a notion of
durable serialisability under relaxed memory
\cite{DBLP:journals/pacmpl/RaadWV19}, but this model does not handle
aborted transactions.

ATLAS \cite{chakrabarti2014atlas} is a software system that provides
durability semantics for 
NVRAM with lock-base multithreaded code. 
The system ensures that the outermost critical sections, which are
protected by one or more mutexes, are failure-atomic by identifying
Failure Atomic Sections (FASEs). These sections ensure that, if at
least one update that occurs to a persistent location inside a FASE is
durable, then all the updates inside the session are
durable. Furthermore, like $\DTML$, ATLAS keeps an persistent undo
log, that tracks the synchronisation operations and persistent stores,
and allows the recovery of rollback FASEs that were interrupted by
crashes. 

Koburn et al. \cite{coburn2011nv} integrate persistent objects into
conventional programs, and furthermore seek to prevent safety bugs
that occur in predominantly persistent memory models, such as multiple
frees, pointer errors, and locking errors. This is done by
implementing NV-heaps, an interface to the NVRAM 
based on ACID transactions that guarantees safety and provides
reasoning about the order 
in which changes to the data structures should become
permanent. 
NV-heaps only handle updates to persistent memory inside transactions
and critical sections. Other systems based on persistent ACID
transactions include Mnemosyne \cite{volos2011mnemosyne}, Stasis
\cite{sears2006stasis} and BerkeleyDB \cite{olson1999berkeley}. 


Ben-David et al. \cite{ben2019delay} developed a system that can
transform programs that consist of read, write and CAS operations in
shared memory, to persistent memory. The system aims to create
concurrent algorithms that guarantee consistency after a fault. This
is done by introducing checkpoints, which record the current state of the
execution and from which the execution can be continued after a
crash. Two consecutive checkpoints form a \emph{capsule}, and if a
crash occurs inside a capsule, program execution is continued from the
previous capsule boundary. We have not applied this technique to
develop \DTML, but it would be interesting to develop and optimise
capsules in an STM context.


Mnemosyne \cite{volos2011mnemosyne} provides a low-level interface to
persistent memory with high-level transactions 
based on TinySTM \cite{felber2008dynamic} and a redo log that is
purposely chosen to reduce ordering constraints.  
The log is flushed 
at the commit of each transaction. As a result, the memory locations
that are written to by the transaction remain unmodified until
commit. Each read operation checks whether data has been modified and
if so, returns the buffered value instead of the value from the
memory. The size of the log increases proportionally to the size of
the transaction, potentially making the checking time consuming.

Hardware based durable transactional memory has also been
proposed~\cite{joshi2018dhtm}  
with hardware support for redo
logging~\cite{joshi2017atom}. 
Other indicative hardware systems help implement atomic durability
either by performing accelerated ordering or by performing the logging
operation are
\cite{nalli2017analysis,lu2014loose}. 


\section{Conclusions}
\label{sec:conc}

In this paper we have defined durable opacity, a new correctness
condition for STMs, inspired by durable
linearizability~\cite{DBLP:conf/wdag/IzraelevitzMS16} for concurrent
objects. The condition assumes a history with crashes such that
in-flight transactions are aborted (i.e., do not continue) after a
crash takes place, and simply requires that the history satisfies
opacity~\cite{2010Guerraoui,GuerraouiK08} after the crashes are
removed. This is a strong notion of correctness but ensures safety for
STMs in the same way that durable
linearizability~\cite{DBLP:conf/wdag/IzraelevitzMS16} ensures safety
for concurrent objects. It is already known that TMS1~\cite{DGLM13},
which is a weaker condition than opacity~\cite{LLM12} is sufficient
for contextual refinement~\cite{DBLP:conf/wdag/AttiyaGHR14}; therefore
we conjecture that durable opacity can provide similar guarantees in a
non-volatile context.  For concurrent objects, more relaxed notions
such as buffered durable
linearizability~\cite{DBLP:conf/wdag/IzraelevitzMS16} have also been
proposed, which requires causally related operations to be committed
in order, but real-time order need not be maintained. Such notions
could also be considered in a transactional memory
setting~\cite{DBLP:journals/pacmpl/DongolJR18}, but the precise
specification of such a condition lies outside the scope of this
paper.

To verify durable opacity, we have developed \DTMS, an operational
characterisation that extends the TMS2 specification with a crash
operation. We establish that all traces of \DTMS are durably opaque,
which makes it possible to prove durable opacity by showing refinement
between an implementation and \DTMS. We develop a durably opaque
example implementation, \DTML, which extends
TML~\cite{DalessandroDSSS10} with a persistent undo log, and
associated modifications such as the introduction of a recovery
operation. 
Finally, we prove durable opacity of
\DTML by establishing a refinement between it and \DTMS. This proof
has been fully mechanised in Isabelle.

Our focus has been on the formalisation of durable opacity and the
development of an example algorithm and verification technique. Future
work will consider alternative implementations of the algorithm, e.g.,
using a persistent set \cite{DBLP:journals/pacmpl/ZurielFSCP19}, or
thread-local undo logs~\cite{izraelevitz2016failure}.  develop and
implement a logging mechanism based on undo and redo log properties
named JUSTDO logging. This mechanism aims to reduce the memory size of
log entries while preserving data integrity after crash occurrences.
Unlike optimistic transactions \cite{chakrabarti2014atlas}, JUSTDO
logging resumes the execution of interrupted FASEs to their last store
instruction, and then executes them until completion. A small log is
maintained for each thread, that records its most recent store within
a FASE, simplifying the log management and reduce the memory
requirements. Future work will also consider weakly consistent memory
models building on existing works integrating persistency semantics
with hardware memory
models~\cite{DBLP:journals/pacmpl/RaadWV19,DBLP:journals/pacmpl/RaadV18}.


\bibliographystyle{splncs04}
\bibliography{references,references2}

\newpage\appendix

\newcommand{\slen}{\mathit{len}}

\section{Soundness of\, $\DTMS$}

We now prove Theorem \ref{theor:dtsm-sound}.
In Lemma \ref{app:lem:wtrace-incl}, we show
that if $h \in traces(\DTMS)$ then $ops(h) \in traces(\TMS)$.
It has already been shown \cite{LLM12} that every trace
of\, $\TMS$ is opaque. Putting these two facts
together, we have that for every trace $h \in traces(\DTMS)$,
$ops(h) \in traces(\TMS)$ is opaque, and so $h$ is durably opaque.

$\TMS$ is fully described in \cite{DGLM13}. We do not present
the automaton explicitly here, but it is {\em precisely}
the the $\DTML$ automaton with the crash action removed.

We prove Lemma \ref{app:lem:wtrace-incl} using an inductive
technique very similar to forward simulation. We exhibit a
relation $R \subseteq states(\DTMS) \times states(\TMS)$,
which we call a weak simulation,
satisfying the properties given in the following definition.
The only difference between this definition and the standard notion 
in Definition\ref{def:for-sim} is that we treat crash events as internal events.
\begin{definition}[Weak simulation]
\label{def:weak-sim}
A \emph{forward simulation} from a concrete IOA $C$ to an abstract IOA
$A$ is a relation $R \subseteq states(C) \times states(A)$ such that 
each of the following holds.  \smallskip

\noindent \emph{Initialisation}. $\forall cs \in start(C).\ \exists as \in start(A).\ R(cs, as)$\smallskip


\noindent \emph{External step correspondence}.

\hfill $\begin{array}[t]{@{}l@{}}
  \forall cs \in reach(C), as \in reach(A), a \in external(C) - \{Crash\}, cs' \in states(C).\ \\
  \qquad R(cs, as) \wedge cs \stackrel{a}{\longrightarrow}_C cs' \imp \exists as' \in states(A).\ R(cs', as') \wedge as \stackrel{a}{\longrightarrow}_A as'
\end{array}
$ \hfill {}\smallskip
\noindent \emph{Internal step correspondence}.

\hfill
$\begin{array}[t]{@{}l@{}}
   \forall cs \in reach(C), as \in reach(A), a \in internal(C) \union \{Crash\}, cs' \in states(C).\ \\
   \qquad R(cs, as) \wedge cs \stackrel{a}{\longrightarrow}_C cs' \imp \\
   \qquad \begin{array}[t]{@{}l@{}}
     R(cs', as) \lor  
     \exists a' \in internal(A), as' \in states(A).\ R(cs', as') \wedge as \stackrel{a'}{\longrightarrow}_A as'
   \end{array}
 \end{array}
 $\hfill{}
\end{definition}
Weak simulations
satisfy the following property.
\begin{lemma}
\label{app:lem:weak-sim-prop}
For any relation $R \subseteq states(\DTMS) \times states(\TMS)$,
if $R$ is a weak simulation from $\DTMS$ to $\TMS$ then
for every $h \in traces(\DTMS)$, $ops(h) \in \TMS$.
\end{lemma}
\begin{proof}
The proof is a simple induction on the length of the executions
of $\DTMS$. More specifically, for each execution of $\DTMS$ $e$, we
inductively construct an execution $e'$ of $\TMS$ such that
$ops(trace(e)) = trace(e')$ which is sufficient. This construction is
entirely standard, and ensures at each step that the final states of each
execution are related by $R$. This guarantees (given the definition of
weak simulation) that we can always extend $e'$ appropriately.
We leave the details to the interested reader.
\end{proof}

We turn now to our main lemma.
\begin{lemma}
\label{app:lem:wtrace-incl}
For every trace $h \in traces(\DTMS)$, $ops(h) \in traces(\TMS)$.
\end{lemma}
\begin{proof}
By Lemma \ref{app:lem:weak-sim-prop}, it is enough to exhibit a weak simulation.
We define our weak simulation $R \subseteq states(\DTMS) \times states(\TMS)$
as follows (we explain the components of this relation shortly):
\begin{align}\label{app:eq:R-def}
(cs, as) \in R \iff \exists i \in \nat. (cs, as) \in G(i) \wedge \forall t. (cs, as) \in L(i, t)
\end{align}
Recall that crash events in $\DTML$ cause $\DTML$'s memory sequence to
be shortened to a length 1 sequence. There is no corresponding event in
$\TML$. Thus, the primary difficulty in our proof is relating the sequence of memories in
states of $\DTML$ with those of $\TML$. The existentially quantified index $i$
in Equation \ref{app:eq:R-def} allows us to do this. Informally,
if $(cs, as) \in R$ then the memory sequence in $c$ is equal to the
suffix of the memory sequence in $a$ beginning at index $i$.

The {\em global relation} $G$, indexed by $i$ is the conjunction of the following::
\begin{align}
\slen(cs.mems) + i &= \slen(as.mems) \label{sim:eq:mems-len}\\
\forall n < \slen(cs.mems). cs.mems(n) &= as.mems(n + i)\label{sim:eq:mems-vals}
\end{align}
The local relation $L$, indexed by $i$ and transaction index $t$ is the
conjunction of the following:
\begin{align}
cs.status_t \notin \{NotStarted&, Committed, Aborted\} \implies\nonumber\\
&cs.beginIndex_t + i = as.beginIndex_t \label{sim:eq:begin-index}
\end{align}
and
\begin{align}
cs.status_t &\neq Aborted \implies cs.status = as.status \label{sim:eq:status}\\
cs.status_t &\neq Aborted \implies cs.rdSet = as.rdSet \label{sim:eq:rdset}\\
cs.status_t &\neq Aborted \implies cs.wrSet = as.wrSet \label{sim:eq:wrset}
\end{align}

We now prove that $R$ is a weak simulation.

\noindent {\bf Initialisation}. Initially, we let $i = 0$.
Because initially $cs.mems = as.mems = [m]$ where $m$ is
the initial memory state, we have 
\begin{align*}
\slen(cs.mems) + i &= \slen(cs.mems) + 0\\
                  &= \slen(as.mems)
\end{align*}
and
\begin{align*}
\forall n < \slen(cs.mems). cs.mems(n) &= as.mems(n + 0)\\
                                     &= as.mems(n)
\end{align*}
and so for initial states $cs, as$, we have $(cs, as) \in G(0)$. Also,
we have $(cs, as) \in L(t, 0)$ as required.

\noindent {\bf External step correspondence}. There are eight cases
to consider. We directly address two. The other cases are very similar.
(Note that we do not treat the $Crash$ action as external, in this weak simulation.)

Let $a = BeginInv_t$ for some thread $t$ and let $cs \stackrel{a}{\longrightarrow}_C cs'$
be a transition of $\DTML$. Let $as$ be an abstract state such that $(cs, as) \in R$,
and let $i$ be the index that witness the existential quantification of $R$.
We first show that the precondition of $BeginInv_t$ is satisfied by $as$.
Note that $cs.status_t = NotStarted$, and thus (by Equation \ref{sim:eq:status}),
we have $as.status_t = NotStarted$, which is sufficient. Because $as$ satisfies
$a$'s precondition, we can let $as'$ be the unique state satsifying 
$as \stackrel{a}{\longrightarrow}_A as'$. It remains to show that $(cs', as') \in R$.
First observe that $(cs', as') \in G(i)$, because $cs'.mems = cs.mems$, $as'.mems = as.mems$
and $(cs, as) \in G(i)$. Furthermore, note that
\begin{align*}
as'.beginIndex_t &= \slen(as.mems)  & \text{Transition relation of $\TML$}\\
                 &= \slen(cs.mems) + i & \text{Equation \ref{sim:eq:mems-len}}\\
                 &= cs'.beginIndex_t + i & \text{Transition relation of $\DTML$}
\end{align*}
as required for Equation \ref{sim:eq:status}. It is easy to check the other conditions
that $(cs', as') \in L(t, i)$.

Let $a = CommitInv_t$ for some thread $t$ and let $cs \stackrel{a}{\longrightarrow}_C cs'$
be a transition of $\DTML$. Let $as$ be an abstract state such that $(cs, as) \in R$,
and let $i$ be the index that witness the existential quantification of $R$.
We first show that the precondition of $CommitInv_t$ is satisfied by $as$.
Note that $cs.status_t = Ready$, and thus (by Equation \ref{sim:eq:status}),
we have $as.status_t = Ready$, which is sufficient. Because $as$ satisfies
$a$'s precondition, we can let $as'$ be the unique state satsifying 
$as \stackrel{a}{\longrightarrow}_A as'$. It remains to show that $(cs', as') \in R$.
First observe that $(cs', as') \in G(i)$, because $cs'.mems = cs.mems$, $as'.mems = as.mems$
and $(cs, as) \in G(i)$. Furthermore, the only local variables that change are
$cs'.status_t$ and $as'.status_t$ so $(cs', as') \in L(t, i)$ is essentially immediate from
$(cs, as) \in L(t, i)$.

\noindent {\bf Internal step correspondence}. There are six cases
to consider. We directly address two. The other cases are very similar.
We first address the $Crash$ action.

Let $a = Crash$ for some thread $t$ and let $cs \stackrel{a}{\longrightarrow}_C cs'$
be a transition of $\DTML$. Let $as$ be an abstract state such that $(cs, as) \in R$,
and let $i$ be the index that witness the existential quantification of $R$.
We must show that $(cs', as) \in R$. It is sufficient to prove that
\begin{align*}
(cs', as) \in G(\slen(as.mems) - 1) \wedge \forall t. (cs', as) \in L(t, \slen(as.mems) - 1)
\end{align*}
So we let $i' = \slen(as.mems) - 1$. Note that $cs'.mems = [last(cs.mems)]$. Thus
\begin{align*}
\slen(as.mems) &= 1 + \slen(as.mems) - 1\\
               &= \slen(cs'.mems) + \slen(as.mems) - 1
\end{align*}
as required for Equation \ref{sim:eq:mems-len}. Also, if $n < \slen(cs'.mems)$ then $n=0$,
and
\begin{align*}
cs'.mems(0) &= last(cs.mems)\\
            &= as.mems(\slen(cs.mems) - 1 + i) & \text{Equation \ref{sim:eq:mems-vals}}\\
            &= as.mems(\slen(as.mems) - 1) & \text{Equation \ref{sim:eq:mems-len}}
\end{align*}
as required for equation \ref{sim:eq:mems-vals}. To prove the local relation,
note that for all $t$, $cs'.status_t \in \{NotStarted, $ $Committed, Aborted\}$
so there is nothing to prove for Equation \ref{sim:eq:status}. Furthermore,
If $cs'.status_t = NotStarted$ then $cs.status_t = NotStarted$ and therefore
the other properties of $L(t, \slen(as.mems) - 1)$ are straightforwardly
maintained.

Now, let $a = DoCommitWriter_t$ for some thread $t$ and let $cs \stackrel{a}{\longrightarrow}_C cs'$
be a transition of $\DTML$. Let $as$ be an abstract state such that $(cs, as) \in R$,
and let $i$ be the index that witness the existential quantification of $R$.
In this case we show that the $\DTML$ transition simulations the $DoCommitWriter_t$
transition in $\TML$. As usual, we show that the precondition holds in $as$.
Again, Equation \ref{sim:eq:status} is enough to prove that the status part of the
precondition holds. We must also show that if $cs.rdSet_t$ is consistent
with respect to $last(cs.mem)$ it is also consistent with respect to
$last(as.mem)$. But
\begin{align*}
last(cs.mem) &= cs.mems(\slen(cs.mems) - 1)\\
             &= as.mems(\slen(cs.mems) - 1 + i)\\
             &= as.mems(\slen(as.mems) - 1)\\
             &= last(as.mems)
\end{align*}
which is sufficent. Clearly, because $cs.wrSet_t = as.wrSet_t$, the
nonemptiness of $cs.wrSet_t$ implies the nonemptiness of $as.wrSet_t$.
Because, $as$ satsifies the precondition of $a$, we let $as'$ be the
unique state satisfying $as \stackrel{a}{\longrightarrow}_A as'$.
It remains to show that $(cs', as') \in R$. To do so, we let $i$ be the
index witnessing this fact, so that we prove
\begin{align*}
(cs', as) \in G(i) \wedge \forall t. (cs', as) \in L(t, i)
\end{align*}
First, observe that
\begin{align*}
\slen(as'.mem) &= \slen(as.mem) + 1\\
               &= \slen(cs.mem) + i + 1\\
               &= \slen(cs'.mem) + i
\end{align*}
as required for Equation \ref{sim:eq:mems-len}. It is straightforward to see
that \ref{sim:eq:mems-vals} is preserved. The only local variables that change
are $cs'.status_t$ and $as'.status_t$, which are both equal to $Committed$,
so $(cs', as') \in L(t, i)$.

Similar arguments prove that the internal step correspondence conditio is met
for the other actions.
\end{proof}

\begin{theorem}[Soundness of $\DTMS$]
Every trace of\, \DTMS is durably opaque.
\end{theorem}
\begin{proof}
Let $h \in traces(\DTMS)$. By Lemma \ref{app:lem:wtrace-incl},
$ops(h) \in traces(\TMS)$ and so $ops(h)$ is opaque \cite{LLM12}. Now, by
Definition \ref{def:duropaque}, $h$ is durably opaque.

\end{proof}


\end{document}